\pdfoutput=1
\documentclass[%
onecolumn,
notitlepage,
reprint,
showpacs,
preprintnumbers,
showkeys,
amsmath,amssymb,
prx,
nofootinbib,
]{revtex4-2}

\newcommand{\mypreamble}{
\title[\shortTITLE]{\TITLE}
\author{Ji\v{r}\'{\i} Proch\'{a}zka}
\email{jiri.prochazka@fzu.cz}
\affiliation{Institute of Physics of the Czech Academy of Sciences,\\ Na Slovance 2, 18221 Prague 8, Czech Republic} 

\date{\today}% It is always \today, today,
}
\usepackage{articleuniv}

\newcommand{\TITLE}{Stochastic state-transition-change process and \\ time resolved velocity spectrometry}
\newcommand{\shortTITLE}{Stochastic STC process and time resolved velocity spectrometry}
\hypersetup{pdfauthor={Jiri Prochazka}, pdftitle={\shortTITLE}}

\graphicspath{{figures/}}

\newkeycommand{\dElementSymbol}[event]{ \ensuremath{ \text{d}(\commandkey{event})}}

\newcommand{\indexset}{I}
\newcommand{\indexsetsequence}{ i \in (0, \dotsc, \ntrans) }
\newcommand{\indexsetsequencemone}{ i \in (0, \dotsc, \ntrans-1)}

\newcommand{\stochproc}{\{\Xw[i=i] : i \in \indexset \}}

\newcommand{\varsymb}{X}
\newcommand{\wvarsymb}{\varsymb}
\newkeycommand{\Xw}[i=]{\wvarsymb_{\commandkey{i}}}
\newkeycommand{\dXw}[i=]{\text{d}\wvarsymb_{\commandkey{i}}}

\newkeycommand{\Xwin}[i=]{\wvarsymb_{\text{in}}^{\commandkey{i}}}
\newkeycommand{\dXwin}[i=]{\text{d}\wvarsymb_{\text{in}}^{\commandkey{i}}}

\newkeycommand{\Xwout}[i=]{\wvarsymb_{\text{out}}^{\commandkey{i}}}
\newkeycommand{\dXwout}[i=]{\text{d}\wvarsymb_{\text{out}}^{\commandkey{i}}}
\newkeycommand{\XwNR}[i=]{\wvarsymb^{NR}_{\commandkey{i}}}
\newkeycommand{\XwinNR}[i=]{\wvarsymb_{\text{in}\ifcommandkey{i}{,{\commandkey{i}}}{}}^{NR}}
\newkeycommand{\XwoutNR}[i=]{\wvarsymb_{\text{out}\ifcommandkey{i}{,{\commandkey{i}}}{}}^{NR}}
\newcommand{\ntrans}{M}  % for a sequence of transitions
\newcommand{\rangeXwNR}[2]{%
   \ifthenelse{ \equal{\detokenize{#1}}{\detokenize{#2}}}{\XwNR[i=#1]}{
		   \ifthenelse{ \equal{\detokenize{#1}}{\detokenize{0}} \AND \equal{\detokenize{#2}}{\detokenize{1}} }{\XwNR[i=#1], \XwNR[i=#2]}{
				   \ifthenelse{ \equal{\detokenize{#1}}{\detokenize{#2}} }{\XwNR[i=#1]}{\XwNR[i=#1], \dotsc, \XwNR[i=#2]}
		   }
   }
}
\newkeycommand{\ifcomma}[i=]{\ifcommandkey{i}{,{\commandkey{i}}}{}}

\newcommand{\statespacesymbol}{S}
\newcommand{\stateelementsymbol}{s}

\newcommand{\stateset}[1]{\statespacesymbol_{#1}}

\newcommand{\statee}[1]{\stateelementsymbol_{#1}}
\DeclareMathOperator{\dos}{dos}

\newcommand{\dosi}[1]{\dos_{#1}}

\newcommand{\dosargSi}[1]{\dos_{#1}(\Xw[i=#1])}

\ExplSyntaxOn
\keys_define:nn { theoryouti }
 {
  i .tl_gset:N = \g_theoryouti_i_tl,
 }

\ProcessKeysOptions{ theoryouti }

\NewDocumentCommand{\stateSoutCondASi}{ O{} } %O{}
{
  \keys_set:nn { theoryouti }
  {
    #1 % evaluate the keys in the optional argument
  }%{}
  \statee{\g_theoryouti_i_tl}\mid \stateset{\g_theoryouti_i_tl}
}
\ExplSyntaxOff

\newkeycommand{\rhoS}[i=]{\rho_{\stateset, \commandkey{i}}}
\newcommand{\ProbTi}[1]{P_{T\ifcomma[i=#1]}}
\newcommand{\ProbTiarg}[1]{\ProbTi{#1}(\Xw[i={#1}])} 
 
\newcommand{\rhoCindexsymb}{C}
\newcommand{\rhoCi}[1]{\rho_{\rhoCindexsymb\ifcomma[i=#1]}}
\newcommand{\rhoCiarg}[1]{\rhoCi{#1}(\Xw[i={#1}], \Xw[i={#1+1}])} %, \rangeXwNR{0}{#1})} 

\newkeycommand{\rhoCindep}[i=]{\rho_{\widetilde{\rhoCindexsymb}\ifcommandkey{i}{,{\commandkey{i}}}{}}}   
\newcommand{\refsasprocessfull}{}

\newcommand{\refsasprocessSTCmainfull}{process-stc-as:independent_realizations,process-stc-as:markov_property}

\newcommand{\refsasprocessSTCextrafull}{process-stc-as:state_prob_transition,\refsasprocessSTCmainfull}

\newcommand{\refsasprocessSTCfull}{\refsasprocessfull,\refsasprocessSTCextrafull}

\newcommand{\refsasprocessSTCsequenceextrafull}{}

\newcommand{\refsasprocessSTCsequencefull}{\refsasprocessSTCfull,\refsasprocessSTCsequenceextrafull}

\newtheorem{definition}{Definition}[section]
\newtheorem{assumption}{Assumption}[section]
\crefname{assumption}{assumption}{assumptions}

\newtheorem{proposition}{Proposition}[section]
\theoremstyle{remark}
\newtheorem{remark}{Remark}[section]

\newcommand{\intervald}[1]{ \ensuremath{ (#1, #1 + \text{d} #1) } }

\newcommand{\refsasprocessSTCsequencetofextra}{as:Xw_tof,as:time_tof,as:velocity_tof,as:force_tof,as:state_space,as:ProbT_tof}

\begin{document}
\mypreamble

\begin{abstract}
%Brownian motion (random walk) is probably the best known example of stochastic motion. 
Motion of particles (bodies) in presence of random effects can be considered stochastic process. However, application of widely known stochastic processes used for description of particle motion is reduced to relatively small class of particle transport phenomena. Stochastic state-transition-change (STC) process is suitable for description of many systems. In this paper it is shown under which assumptions formulae of time resolved velocity spectrometry can be derived with the help of STC process. It opens up new possibilities of unified description of particles moving in a force field in presence of random effects. It extends possibilities of applications of theory of stochastic processes in physics. 
\end{abstract}

% insert suggested PACS numbers in braces on next line
%\pacs{13.85.Dz,13.85.Lg,14.20.Dh}
% insert suggested keywords - APS authors don't need to do this
\keywords{stochastic state-transition-change process, stochastic particle motion, time resolved velocity spectrometry, particle transport}
%\keywords{light, polarization, polarizer, light-matter interaction, properties of light, laser, }

\maketitle

\clearpage
\tableofcontents

\section{\label{sec:introduction}Introduction}
Motion of particles, or any other bodies, under various conditions has been studied in physics for a very long time. In some cases the motion can be described deterministically using, e.g., Newton's second law of motion. In some other cases it is necessary to take into account several random effects. Brownian motion (random walk) is well known example of stochastic process. Stochastic cyclotron motion \cite{Lemons2002} is another example of stochastic process. It introduces randomness to motion of a charged particle in a magnetic field by adding dissipation and fluctuation terms to the corresponding deterministic equation of motion (the terms can be added to any ordinary differential equation having time as an independent variable).

However, widely known stochastic processes are not suitable for description of all particle transport phenomena in presence of random effects. Stochastic state-transition-change (STC) process introduced in \cref{process-stc-sec:model_stochastic_process} in \cite{Prochazka2022statphys_process_stc} extends possibilities of descriptions of motion of particles which may have random initial properties (states), may or may not reach given position, and may or may not change their properties during transport. It will be shown in this paper under which conditions (assumptions) it is possible to derive main formulae of time resolved velocity spectrometry using STC process.

This paper is structured as follows. Derivation of formulae of widely know time resolved velocity spectrometry with the help of stochastic STC process is in \cref{sec:model}. The spectrometry is one of well known time-of-flight (TOF) methods which allows determination of spectrum of emitted particles from a source as a function of emission time and velocity on the basis of experimental data, see \cref{sec:data_analysis}. The possibilities of generalization of the time resolved velocity spectrometry with the help of STC stochastic process are discussed in \cref{sec:generalization}. Concluding remarks are in \cref{sec:conclusion}.

\section{\label{sec:model}Probability model - time resolved velocity spectrometry}
\subsection{Stochastic process}

\newcommand{\speed}{v}

\newcommand{\vmin}{v_{\text{min}}}
\newcommand{\vmax}{v_{\text{max}}}
\newcommand{\tmini}[1]{t_{#1}^{\text{min}}}
\newcommand{\tmaxi}[1]{t_{#1}^{\text{max}}}

%\subsubsection{\label{sec:example_tof_tv}Time resolved velocity spectrometry}
%The simplest particle transport phenomenon is motion of particles in free space. 

Stochastic STC process can describe particle motion of particles in a force field when initial conditions are characterized by probability (density) functions. 

Consider a source emitting particles of different speeds in the same direction and at different times as an example of nontrivial particle motion. It may not be possible to detect the particles at the place where they are emitted but it may be possible to measure some quantities characterizing the particle transport at several distances $x_i$ from the source ($i \in (0, ..., \ntrans)$, $\ntrans >  0$, $x_{i} < x_{i+1}$ and the spatial $x$-axis has the same orientation as the direction of the velocities of the particles). One may ask how to determine the characteristics of the emitted particles at the place where they are emitted on the basis of quantities which can be measured.

%To describe this transport phenomenon one may  $\surfacein$ at position $x_{i}$ through which the emitted particles passed. The control surfaces may be taken as flat and parallel to each other. One may use Cartesian coordinate system with the $x$-axis being orthogonal to the control surfaces (the control surface $\surface_{i}$ being at position $x_{i}$). The source point may be chosen as part of the surface $\surface_{0}$ (i.e., being at position $x_{i}$)\footnote{Instead of the two control surfaces two control points could be considered, too.}.
\begin{assumption}[Time interval of emitted particles]
\label{as:time_tof}
The particles were emitted in a burst in time interval from $\tmini{0}$ to $\tmaxi{0}$ ($\tmini{0} < \tmaxi{0}$) at position $x_0$. % x_{i} 
\end{assumption}
\begin{assumption}[Velocity of emitted particles]
\label{as:velocity_tof}
All the emitted particles emitted at position $x_0$ had the same direction of velocity. The particles could have different values of speeds $v_0$ in the interval from $\vmin$ to $\vmax$ ($\vmin \leq \vmax$).
\end{assumption}
\begin{assumption}[Constant speed, zero force]
\label{as:force_tof}
The speeds of individual particles did not change during transport from position $x_{i}$ to $x_{i+1}$. I.e., no force acted on the particles during the transport (considering only inertial motion).
%Therefore, it holds for speed $v$ of an emitted particle.
%\begin{align}
%		v &= \vinx = \voutx  \label{eq:tof_speed} 
%\end{align}
\end{assumption}
\begin{assumption}[Variables]
\label{as:Xw_tof}
Let particle at position $x_{i}$, at time $x_{i}$ have speed $v_{i}$. Let % Let an output particle at position $x_{i+1}$, time $t_{i+1}$ have velocity $v_{i+1}$ (see \cref{as:force_tof}), i.e., let 
\begin{align} % ODO: solve the problems with variables...
		%\Xwin    &= (v, t_{i})   \\ 
		%\XwinNR  &= (x_{i})         \\ 
		%\Xwout   &= (v, t_{i+1}) \\ 
		%\XwoutNR &= (x_{i+1})  \, .
		\Xw[i=i]  &= (v_{i}, t_{i})   \\ 
		\XwNR[i=i]&= (\vmin, \vmax, \tmini{0}, \tmaxi{0}, x_0,\dotsc,x_{i})   \, .
		%\Xwout   &= (v, t_{i+1}) \\ 
		%\XwoutNR &= (x_{i+1})  \, .
\end{align}
for all $\indexsetsequencemone$. I.e., $v_{i}$ and $t_{i}$ are considered random variables and the other parameters are non-random variables (only some of them may or may not be explicitly written as arguments of functions in the following).
% ODO: explain why	$\XwNR[i=i]$ depends on $x_0$ and not on other non-random parameters
%Let the variables $x_{i}$ and $x_{i+1}$ be non-random variables and let the other variables be random.
\end{assumption}
\begin{assumption}[State spaces]
\label{as:state_space}
Let state space $\stateset{i}$ contain states of particles, represented by variables $\Xw[i=i]$, when they pass through position $x_i$.
\end{assumption}
\begin{remark}
\Cref{as:state_space} implies that number of states of system which were in any state in state space $\stateset{i}$ is the same as number of particles which passed through $x_i$.
\end{remark}

\begin{assumption}[Probability of transition]
\label{as:ProbT_tof} 
It holds (for all $\indexsetsequencemone$)
\begin{equation}
    \ProbTiarg{i} = 1  \, ,
\label{eq:ProbT_tof}
\end{equation}
i.e., a particle at position $x_i$ always moved to position $x_{i+1}$, independently on value of $\Xw[i=i]$ .
\end{assumption}
%
%
%\begin{assumption}
%\label{as:statesets_tof}
%Let $\stateset{i}$ (resp.~$\stateset{i+1}$) contain all states of particles corresponding to variables $\Xw[i=i]$ (resp.~$\Xw[i=i+1]$).
%\end{assumption}
%

%\newcommand{\tofXwiargs}[i]{x_{#1}, v, t_{#1}}
\newcommand{\tofXwinargs}{x_{i}, v_{i}, t_{i}}
\newcommand{\tofXwvintegargs}[1]{x_{#1}, t_{#1}}

\newcommand{\dosargitof}[1]{\dos_{#1}(x_{#1}, v_{#1}, t_{#1})}
\newcommand{\dosargvintegtof}[1]{\dos_{#1}^{v}(\tofXwvintegargs{#1})}
\newcommand{\dositofX}[2]{\dosi{#1}(x_{#1},  v_{#1}=v_{#2}, t_{#1}=t_{#2} - \frac{x_{#2} - x_{#1}}{v_{#2}} )}

\begin{definition}
		If \cref{as:Xw_tof} holds then the density of states (DOS) $\dosargSi{i}$ defined by \cref{process-stc-eq:notation_dosargSi} in \cite{Prochazka2022statphys_process_stc} can be written also as 
\begin{alignat}{3}
	%	\dosargitof{i}            &= \dosargSi{i}          && \quad\quad\quad \indexsetsequence         \label{eq:dosargi_tof}\\
		\dosargitof{i}            &= \dosargSi{i}           && \quad\quad\quad          \label{eq:dosargi_in_tof}
	%	 N_{i}                     &= \Ni{i}                && \quad\quad\quad \indexsetsequence        \label{eq:Ni_tof} \\
	%	\rhoSargitof{i}           &= \rhoargSi{i}          && \quad\quad\quad \indexsetsequence        \label{eq:rhoSarg_tof}
\end{alignat}
for all $\indexsetsequence$. %The newly defined functions are on the left-hand sides of the equations. % on non-random variables (i.e., on rotation angles of polarizers) may or may not be written explicitly. 
\end{definition}
\begin{definition}
		Let $\dosargvintegtof{i}$ be defined by ($\indexsetsequence$)
\begin{align}
		\dosargvintegtof{i}     &= \int_{v_i} \dosargitof{i} \text{d} v_i  \, .    \label{eq:dosargvintegtof}
\end{align}
\end{definition}

\newcommand{\rhoCargitof}{\rhoCi{i}(x_{i}, v_{i}, t_{i}, x_{i+1}, v_{i+1}, t_{i+1})}
\newcommand{\rhoCargispeedfulltof}{\rhoCi{i}(x_{i}, v_{i}, t_{i}, x_{i+1}, v_{i+1}=v_{i}, t_{i+1})}
\newcommand{\rhoCargispeedtof}{\widetilde{\rho}_{\rhoCindexsymb\ifcomma[i=i]}(x_{i}, v_{i}, t_{i}, x_{i+1}, t_{i+1})}

\begin{definition}
		If \cref{as:Xw_tof} holds then function $\rhoCiarg{i}$ defined by \cref{process-stc-eq:rhoCiarg} \cite{Prochazka2022statphys_process_stc} can be written also as
\begin{align}
		%\rhoCargitof    &= \rhoCi{i}(\Xw[i=i],\Xw[i=i+1])         & i&=0,...,\ntrans-1 \, . \label{eq:tof_rhoCargpol}
		\rhoCargitof     &=  \rhoCiarg{i}       & & \, . \label{eq:tof_rhoCargpol}
\end{align}
where non-random variables $x_{i}$ and $x_{i+1}$ have been written explicitly, and $\indexsetsequencemone$. The function $\rhoCiarg{i}$ has meaning of probability function that a particle at position $x_i$ of properties characterized by random variables $\Xw[i=i]$ had properties characterized by random variables $\Xw[i=i+1]$ at position $x_{i+1}$.
\end{definition}

\begin{definition}
		Let the probability density function $\rhoCiarg{i}$ satisfying the assumption of constant speed of a particle (see \cref{as:velocity_tof,as:force_tof}) be denoted as $\rhoCargispeedtof$, i.e., it holds (for all $\indexsetsequence$)
\begin{equation}
     \rhoCargispeedtof = \rhoCargispeedfulltof \, .
\end{equation}
%where 
%\begin{equation}
%v = v_{i} = v_{i+1}  \, .
%\end{equation}
\end{definition}

\begin{definition}[Minimal and maximal particle arrival time]
		The last (resp.~the first) particle which gets to $x_{i+1}$ ($\indexsetsequencemone$) is particle having the lowest (resp.~the highest) speed $\vmin$ (resp.~$\vmax$) which was at $x_{i}$ at time $\tmaxi{i}$ (resp.~$\tmini{i}$).
\end{definition}

\begin{definition}
\label{def:process_STC_tof}
%\letbeprocess{process-stc-def:process_STC_sequence}
Let $\stochproc$ be stochastic process given by \cref{process-stc-def:process_STC_sequence} in \cite{Prochazka2022statphys_process_stc}.
 I.e., it satisfies \cref{\refsasprocessSTCsequencefull} in \cite{Prochazka2022statphys_process_stc}. Let it satisfy also \cref{\refsasprocessSTCsequencetofextra}.
\end{definition}

\subsection{\label{sec:formulae}Derivation of several formulae}
The following statements will be derived assuming stochastic process given by \cref{def:process_STC_tof}, if not mentioned otherwise.

\newcommand{\refsaslinetrajectory}{as:velocity_tof,as:force_tof,as:Xw_tof} %,as:statesets_tof}
\begin{proposition}
\label{thm:trajectory}
%\letbeprocess{def:process_STC_tof}
		Particle of constant speed $\speed$ at position $x_{i}$ and at time $t_{i}$ reached position $x_{i+1}$ at time $t_{i+1}$ which is equal to (for all $\indexsetsequencemone$)
\begin{align}
		t_{i+1} = t_{i} + \frac{x_{i+1} - x_{i}}{\speed} \, .  \label{eq:trajectory}
\end{align}
It holds  (for all $\indexsetsequence$)
\begin{align}
		t_{i} = t_0 + \frac{x_{i} - x_0}{\speed} \, .  \label{eq:trajectory_full}
\end{align}
\end{proposition}
\begin{proof}
It follows from \cref{as:velocity_tof,as:force_tof} and Newton's second law of motion. % imply \cref{eq:trajectory}.
\end{proof}
\begin{proposition}
\label{thm:tmini_tmaxi}
%\letbeprocess{def:process_STC_tof}
%Let \cref{as:time_tof,\refsaslinetrajectory} hold. 
 %The first (resp.~the last) particle reaches $x_{i+1}$ at time $\tmini{i+1}$ (resp.~$\tmaxi{i+1}$). 
It holds (for all $\indexsetsequencemone$)
\begin{align}
		\tmini{i+1} &= \tmini{i} + \frac{x_{i+1} - x_{i}}{\vmax}      \label{eq:tmini}\\
		\tmaxi{i+1} &= \tmaxi{i} + \frac{x_{i+1} - x_{i}}{\vmin} \, . \label{eq:tmaxi}
\end{align}
and (for all $\indexsetsequence$)
\begin{align}
		\tmini{i} &= \tmini{0} + \frac{x_{i} - x_0}{\vmax}      \label{eq:tmini_full}\\
		\tmaxi{i} &= \tmaxi{0} + \frac{x_{i} - x_0}{\vmin} \, . \label{eq:tmaxi_full}
\end{align}
\end{proposition}
\begin{proof}
It follows from \cref{as:time_tof,\refsaslinetrajectory} and \cref{eq:trajectory,eq:trajectory_full}.
\end{proof}

\begin{proposition}
$t_{i} \in [\tmini{i},\tmaxi{i}]$ is equivalent to $v_{i} \in [\vmin,\vmax]$ (for all $\indexsetsequence$).
\end{proposition}
\begin{proof}
It follows from \cref{eq:tmini_full,eq:tmaxi_full}.
\end{proof}

\newcommand{\toftrajectoryspeed}{t_{i} - t_{i+1} + \frac{x_{i+1} - x_{i}}{v_{i}}}
\newcommand{\rhoCargideltatrajectoryspeedtof}{\delta(\toftrajectoryspeed)}
\newcommand{\rhoCargideltaspeedtof}{\rhoCargideltatrajectoryspeedtof}
% ODO: generalize the deterministic case, if $\Xw[i=i+1]$ is function of $\Xw[i=i]$ (denoted as $\Xw[i=i+1](\Xw[i=i])$) then $\rhoargSi{i} = delta(\Xw[i=i+1](\Xw[i=i]))$ ...
\begin{proposition}
\label{thm:rho_C_tof}
%If the variables $\Xw[i=i+1]$ are determined uniquely for given index $i$ and fixed variables $\Xw[i=i]$ and there is continuous mapping \Xw[i=i](u)
%\letbeprocess{def:process_STC_tof}
%Let it satisfy \cref{\refsaslinetrajectory}. 
 It holds  (for all $\indexsetsequencemone$)
\begin{align}
\rhoCargispeedtof      &=  \rhoCargideltaspeedtof \, . \label{eq:rho_C_tof} 
\end{align}
\end{proposition}
\begin{proof}
%It follows from \cref{\refsaslinetrajectory} that particle in given state $\Xw[i=i]$ transitioned to given state $\Xw[i=i+1]$ if and only if
It follows from \cref{as:velocity_tof,as:force_tof} and \cref{eq:trajectory}.
\end{proof}

\begin{proposition}[Transformation of DOS]
\label{thm:tof_transition_dos}
%\letbeprocess{def:process_STC_tof}
It holds (for all $\indexsetsequencemone$)
\begin{align}
        %\dosargouttof
		\dosi{i+1}&(x_{i+1},  v_{i+1}, t_{i+1}) \nonumber\\
		& = \begin{cases}
				\dosi{i}(x_{i},  v_{i}=v_{i+1}, t_{i}=t_{i+1} - \frac{x_{i+1} - x_{i}}{v_{i+1}} ) \\ & \mkern-216mu \text{if } t_{i+1} \in [\tmini{i+1}, \tmaxi{i+1}] \text{ and } v_{i+1} \in [\vmin, \vmax] \\
				0 & \text{otherwise}
		\end{cases}
		%\, .
        \label{eq:dos_tof_increment}
\end{align}
where $\tmini{i+1}$ and $\tmaxi{i+1}$ can be determined using \cref{thm:tmini_tmaxi}. Equivalently, it holds (for all $\indexsetsequence$)
\begin{align}
        %\dosargouttof
		\dosi{i}&(x_{i},  v_{i}, t_{i}) \nonumber\\
		& = \begin{cases}
				\dosi{0}(x_0,  v_0=v_{i}, t_0=t_{i} - \frac{x_{i} - x_0}{v_{i}} ) \\ & \mkern-216mu  \text{if } t_{i} \in [\tmini{i}, \tmaxi{i}] \text{ and } v_{i} \in [\vmin, \vmax] \\
				0 & \text{otherwise}
		\end{cases}
		%\, .
        \label{eq:dos_tof_full}
\end{align}
where $\tmini{i}$ and $\tmaxi{i}$ are given by \cref{eq:tmini_full,eq:tmaxi_full}. 
\end{proposition}
\begin{proof}%[Proof 2]
According to \cref{as:velocity_tof,as:force_tof} it must hold $v_{i}=v_{i+1}$. The number of particles in interval $\intervald{v_{i}}\times\intervald{t_{i}}$ at position $x_{i}$ and time $t_{i}$ divided by $\text{d}v_{i}\text{d}t_{i}$ is equal to $\dosi{i}(x_{i},  v_{i}, t_{i})$ (see definition of DOS given by \cref{process-stc-eq:def_dos} in \cite{Prochazka2022statphys_process_stc}). The particles reach $x_{i+1}$ at time $t_{i+1}$ given by \cref{eq:trajectory}. It implies . $\dosi{i+1}(x_{i+1},  v_{i+1}, t_{i+1}) =0 $ in regions of $t_{i+1}$ and $v_{i+1}$ which are outside physical region. %\Cref{eq:dos_tof_increment} is equivalent to \cref{eq:dos_tof_full}.  
The equivalence of \cref{eq:dos_tof_increment,eq:dos_tof_full} can be proven using \cref{thm:trajectory}.
\end{proof}
\begin{remark}
		\Cref{thm:tof_transition_dos} can be derived in another way using \cref{process-stc-thm:transition_dos} in \cite{Prochazka2022statphys_process_stc} and \cref{eq:ProbT_tof} and function $\rhoCargitof$ expressed as a 2-dimensional delta function corresponding to \cref{\refsaslinetrajectory}. To work with $n$-dimensional delta functions is, however, in general more delicate than in 1-dimensional case.
\end{remark}

\newcommand{\integspeedlimits}{\int^{\vmax}_{\vmin}}

\newcommand{\integspeedlimitsfull}{\int^{\vmax}_{\vmin}}

\begin{proposition}
\label{thm:dosargvintegtof_given_by_integration}
%\letbeprocess{def:process_STC_tof} 
 It holds
\begin{align}
		\dosargvintegtof{i}
		& = \begin{cases}
		\integspeedlimitsfull
        \dosargitof{i}
		\text{d} v_{i}
		&  \text{if } t_{i} \in [\tmini{i}, \tmaxi{i}] 
		\label{eq:dosargvintegtof_given_by_integration} \\
		0 
		& \text{otherwise}
		\end{cases}
\end{align} 
where $\tmini{i}$ and $\tmaxi{i}$ are given by \cref{eq:tmini_full,eq:tmaxi_full}. 
\end{proposition}
\begin{proof}
Insertion of $\dosi{i}(x_{i}, v_{i}, t_{i})$ given by \cref{eq:dos_tof_increment} into \cref{eq:dosargvintegtof} implies \cref{{eq:dosargvintegtof_given_by_integration}}. 
%Given time $t_{i}$ then speed of particle $\speed$ at $x_{i}$ is equal to
%\begin{align}
%		\speed = \frac{x_{i} - x_0}{t_{i} - \tmini{0}}.
%\end{align}
%		If $\tmini{0} < t_{i} < \tmini{i}$ then $\speed > \vmax$ (see \cref{eq:tmini_full}). If $t_{i} > \tmaxi{i}$ then $\speed < \vmin$ (see \cref{eq:tmaxi_full}). It implies the integration limits in \cref{eq:dosargvintegtof_given_by_integration}. Equivalently, one could integrate from $\vmin$ to $\vmax$ if $t_{i} \in [\tmini{i},\tmaxi{i}]$ and $\dosargvintegtof{i}=0$
\end{proof}

\begin{proposition}
\label{thm:tof}
%\letbeprocess{def:process_STC_tof} 
% I.e., \cref{\refsasprocessSTCsequencetof} hold. 
		It holds (for all $\indexsetsequencemone$)
\begin{align}
		&\dosargvintegtof{i+1}
		= \noindent \\
		&\qquad \int_{\vmin}^{\vmax}
		\int_{\tmini{i}}^{\tmaxi{i}}   
        \dosargitof{i}
		\rhoCargideltatrajectoryspeedtof
		\text{d} t_{i} 
		\text{d} v_{i}  
		\label{eq:dos_tof_vlad1}
		\, ,
\end{align} 
if $t_{i+1} \in [\tmini{i+1},\tmaxi{i+1}]$ (see \cref{eq:tmini,eq:tmaxi}), otherwise $\dosargvintegtof{i+1}=0$.
\end{proposition}
\begin{proof}
Let us consider $\Xw[i=j]=(x_{j},v_{j},t_{j})$ for $j \in (0,\dotsc,i)$ and $\Xw[i=i+1]=(x_{i+1},t_{i+1})$ for fixed index $\indexsetsequencemone$.
		\Cref{process-stc-thm:transition_dos} in \cite{Prochazka2022statphys_process_stc} and \cref{eq:ProbT_tof,eq:rho_C_tof} imply \cref{eq:dos_tof_vlad1}. 
\end{proof}

\begin{proposition}
\label{thm:dosargvintegtof_given_by_integration_2}
%\letbeprocess{def:process_STC_tof} 
 It holds (for all $\indexsetsequencemone$)
\begin{align}
		&\dosargvintegtof{i+1}
		 = \noindent \\
		&\qquad\integspeedlimits
		\dositofX{i}{i+1}
		\text{d} v_{i+1}
		\, ,
        \label{eq:dos_tof_vlad2}
\end{align} 
		if $t_{i+1} \in [\tmini{i+1},\tmaxi{i+1}]$, otherwise $\dosargvintegtof{i+1}=0$. Equivalently (for all $\indexsetsequence$),
\begin{align}
		\dosargvintegtof{i}
		& = 
		\integspeedlimits
		\dositofX{0}{i}
		\text{d} v_{i}
		\, ,
        \label{eq:dos_tof_vlad2_full}
\end{align} 
if $t_{i} \in [\tmini{i},\tmaxi{i}]$, otherwise $\dosargvintegtof{i}=0$.
\end{proposition}
\begin{proof}[Proof 1]
By integrating \cref{eq:dos_tof_increment} over $v_{i+1}$ and using \cref{eq:dosargvintegtof_given_by_integration} one obtains \cref{eq:dos_tof_vlad2}. By integrating \cref{eq:dos_tof_full} over $v_{i}$ and using \cref{eq:dosargvintegtof_given_by_integration} one obtains \cref{eq:dos_tof_vlad2_full}. The equivalence of \cref{eq:dos_tof_vlad2,eq:dos_tof_vlad2_full} can be proven using \cref{thm:tof_transition_dos,thm:trajectory}.
\end{proof}
\begin{proof}[Proof 2]
Performing the integration over $t_{i}$ in \cref{eq:dos_tof_vlad1} implies \cref{eq:dos_tof_vlad2} ($v_{i}$ in \cref{eq:dos_tof_vlad1} and $v_{i+1}$ in \cref{eq:dos_tof_vlad2} are only integration variables, they can be renamed).
%		A particle can be at $x_{i+1}$ at given time $t_{i+1}$ if and only if its speed is in $[\vmin, \vmax]$. If this is not the case then $\dosargvintegtof{i}$ is zero. This allows to modify the limits of integration over $v_{i}$ in \cref{eq:dos_tof_vlad1} and to obtain \cref{eq:dos_tof_vlad2}.
\end{proof}
\section{\label{sec:data_analysis}Analysis of experimental data}
%\begin{remark}
%\label{rmk:tof_method}
Time dependent densities of states $\dosargvintegtof{i}$ can be measured at several positions $x_{i}$. Unknown parameters $\vmin$, $\vmax$, $\tmini{0}$ and $\tmaxi{0}$, and unknown function $\dosargitof{0}$ (time resolved velocity spectrum of emitted particles) can be determined on the basis of the measured data with the help of formulae derived in \cref{sec:formulae} (see mainly \cref{thm:tof_transition_dos,thm:dosargvintegtof_given_by_integration,thm:dosargvintegtof_given_by_integration_2,thm:tmini_tmaxi}), and (constrained) optimization techniques as discussed in \cref{process-stc-sec:data_analysis_guidelines} in \cite{Prochazka2022statphys_process_stc}.
%\end{remark}

%\begin{remark}
%\label{rmk:tof_method_context}
If we put $M=1$, $x_{i}=0$, $x_{i+1}=x$, $t_{i}=0$ $t_{i+1}=t$, $\vmin=v_1$, $\vmax=v_2$, $\tmini{i}=0$ and $\tmaxi{i}=\Delta T$ then \cref{eq:dos_tof_vlad1,eq:dos_tof_vlad2} are equivalent to eq.~(2) in \cite{Vlad1984} where an extension of the time-of-flight (TOF) method for determination of the time resolved velocity spectrum of particles emitted in intense bursts has been presented for the first time. Many useful comments to the time resolved velocity spectrometry method are in \cite{Vlad1984} (including numerical solutions and tests). 

In \cite{Rezac2012} this method is called \emph{extended TOF method} to distinguish it from \emph{basic TOF method} in which velocity spectrum of particles is determined independently on the emission time. It is shown in \cite{Rezac2012} under which conditions the former method is reduced to the later one (in the case of relatively short intense burst in comparison to the time of flight of particles from source to a detector). The basic TOF method provides less information, but it is significantly easier to use it from both the experimental and data analysis point of view (it is sufficient to use only one detector in sufficiently large distance from the source ensuring that $\Delta T$ is much smaller than the time needed by an emitted particle to travel from the source to the detector).

		Both the types of TOF methods are widely known and have been adapted and successfully applied in various experiments. E.g., with the help of the extended TOF method time resolved neutron energy spectra from D(d,n)$^3$He fusion reactions were determined in \cite{Rezac2012} using analog Monte Carlo reconstruction method (AMCRT). In sect.~2.2.1 in \cite{Rezac2012} several other existing reconstruction methods (algorithms) are mentioned. Efficiency of different methods depends on several factors including the dependence of the density of states $\dosargitof{0}$ (i.e., also measured densities of states $\dosargvintegtof{i}$) which one is trying to determine.

The basic TOF method was applied, e.g., to experimental data of electrons and ions emitted by laser-produced plasmas in \cite{Krasa2007,Krasa2013,Krasa2018,Krasa2020}; further details to TOF spectra for mapping of charge density of ions produced by laser are in \cite{Krasa2014}. 
%\end{remark}

\section{\label{sec:generalization}Generalization of time-of-flight methods}
%\begin{remark}
%Instead of \cref{as:velocity_tof,as:force_tof} one can assume that the input particles had different not only the magnitude of velocities but also the directions of velocities, and that both the directions and the magnitudes of velocities could change during the transport. 
%Instead of \cref{as:force_tof} non-zero external force action on individual particles during the particle transport. %\Cref{as:Xw_tof} must be then modified correspondingly. Instead of \cref{eq:trajectory} one need to determine an output particle state on the basis of its input state (i.e., its initial position and velocity) and, e.g., the second law of motion (law of force). \Cref{eq:rho_C_tof} can be then modified correspondingly. 
 %With the help of \cref{process-stc-thm:transition_dos} in \cite{Prochazka2022statphys_process_stc} 
%\end{remark}

%\begin{remark}
It has been shown in \cref{sec:model} that with the help of stochastic STC process introduced in \cref{process-stc-sec:model_stochastic_process} in \cite{Prochazka2022statphys_process_stc} leads to the well known time-of-flight (TOF) method for determination of time resolved velocity spectrum of emitted particles from a source (TOF-TV). There are other experimental techniques based on measurement of TOF, such as the time-of-flight mass spectrometry (TOF-MS). This method uses an electric field of known strength to determine mass-to-charge ratio of particles (ions). An electric and magnetic fields of known strength are commonly used (in various configurations) to determine mass and electric charge of charged particle by measuring and analysing trajectories of the particles with the help of the Lorenz force. All time-of-flight methods and many methods of mass spectrometry have something in common. They concern particle transport phenomena under various conditions and their aim is to determine properties of the particles. 

Each of the methods is typically suitable for determination of partial spectra being functions of only some variables characterizing properties of particles emitted from a source such as mass, electric charge, velocity (energy) or time of their emission from the source. E.g., the TOF-TV method allows determination of time resolved velocity spectrum (function of ``only'' two random variables characterizing the emitted particles). With the help of stochastic STC process it is possible to generalize the formulae in \cref{sec:formulae} (TOF-TV method) by taking into account:
\begin{enumerate}
\item{non-zero external force}
\item{general initial and final positions and velocities of particles}
\item{various properties of particles such as mass, electric charge, etc.}
\end{enumerate}
Application of this \emph{generalized} TOF-TV method to data can be more complicated than application of the TOF-TV method to data (it may be necessary to consider more random variables). It requires more experimental information. One can take advantage of various experimental methods determining the ``partial'' spectra (integrated over some of the variables) to constrain the ``full'' spectrum, see general guidelines in \cref{process-stc-sec:data_analysis_guidelines} in \cite{Prochazka2022statphys_process_stc}. One can then derive a more general equation than \cref{eq:dos_tof_vlad2}. This allows to study forces acting on particles and their properties (some of them may be specified by random variables). Inertial mass increase in dependence on velocity may be also studied on the basis of experimental data under these conditions \cite{Lokajicek2017hamiltonian}.
%Secondly, the determination of the full spectrum may depend on its shape (functional form). 
%In some cases it may be very difficult to find a convenient parameterization of the full spectrum and fit the values of the free parameters to all the measured data. % (see the optimization procedure discussed in \cref{sec:data_analysis_guidelines}).
%\end{remark}

%\begin{remark}
%\end{remark}

\section{\label{sec:conclusion}Conclusion}
Stochastic STC process introduced in \cref{process-stc-sec:model} in \cite{Prochazka2022statphys_process_stc} can significantly help to improve existing or develop new techniques of measurement of properties of particles moving in an external force field and being specified by random variables. The measurement is essential for characterization of various sources of emitting particles (see, e.g., laser-produced plasmas mentioned in \cref{sec:data_analysis}, or development of deuterium z-pinch as a powerful source of multi-MeV ions and neutrons \cite{Klir2016}). The sources of known properties can be used for various applications. The new types of sources of emitted particles place extra demands on particle detectors (see, e.g., design of a scintillator calorimeter for laser-plasma characterization, or magnetic electron spectrometer spectrometer \cite{Krupka2021}).  %sources can be used for various applications. %The probabilistic model can help with calibration of some types of particle detectors, such as calorimeters. This calibration can be then compared to other calibration methods, which are often based on a Monte Carlo simulation, see, e.g., \cite{Istokskaia2021}. The probabilistic model \Cref{Krupka2021}

With the help of stochastic STC process it is possible to describe in a unified way motion of particles in an external force field and many other particle transport phenomena which looks very distinct at first glance, such as transmission of light through sequence of polarizers \cite{Prochazka2022statphys_polarizers}. Several other applications of stochastic STC process for description of (physical) systems are discussed in \cite{Prochazka2022statphys_process_stc}.

%\bibliography{../../references} 
%apsrev4-2.bst 2019-01-14 (MD) hand-edited version of apsrev4-1.bst
%Control: key (0)
%Control: author (8) initials jnrlst
%Control: editor formatted (1) identically to author
%Control: production of article title (0) allowed
%Control: page (0) single
%Control: year (1) truncated
%Control: production of eprint (0) enabled
%

\end{document}